%% file: cdmz-31-10-14.tex
\documentclass[UKenglish]{eptcs}

\usepackage[usenames,dvipsnames]{color}

\usepackage{times}
\usepackage[latin1]{inputenc}
\usepackage{pslatex}
\usepackage{amsmath}
\usepackage{amssymb}
\usepackage{amsfonts}
\usepackage{amsthm}
\usepackage{hyperref}
\usepackage{breakurl}
\usepackage{prooftree}
\usepackage{txfonts}
\usepackage{lscape}
\usepackage[all]{xy}
\usepackage{textcomp}
\usepackage{breakurl}

\usepackage{color}

\input{macroL}

\providecommand{\event}{ITRS 2014} 

\title{On Isomorphism of  ``Functional''\\ Intersection and Union Types\footnote{
 This work was partially supported by EU Collaborative project  ASCENS 257414, ICT COST Action IC1201 BETTY, MIUR PRIN Project CINA Prot. 2010LHT4KM and Torino University/Compagnia San Paolo Project SALT. 
}}
\author{Mario Coppo\quad
Mariangiola Dezani-Ciancaglini\quad
Ines Margaria\quad
Maddalena Zacchi
\institute{Dipartimento di Informatica Universit\`a di Torino,
corso Svizzera 185, 10149 Torino, Italy}
}

\def\titlerunning{On Isomorphism of  ``Functional'' Intersection and Union Types
}
\def\authorrunning{M. Coppo,  M. Dezani-Ciancaglini, I. Margaria \& M. Zacchi 
}
\begin{document}

\maketitle

%

\begin{abstract}
 Type isomorphism is useful for retrieving library components, since a function in a library can have a type different from, but isomorphic to,  the one expected by the user. Moreover type isomorphism gives for free the coercion required to include the function in the user program with the right type. The present paper faces the problem of type isomorphism in a system with intersection  and union types. In the presence of intersection and union, isomorphism is not a congruence and cannot be characterised in an equational way. A characterisation can still be given, quite complicated by the interference between functional and non functional types. This drawback is faced in the paper by interpreting each  atomic type as the set of functions mapping any argument into the interpretation of the type itself. This choice has been suggested by the initial projection of Scott's inverse limit $\lambda$-model.
The main result of this paper is a condition assuring type isomorphism, based on
an isomorphism preserving reduction.

\end{abstract}

\section{Introduction}

In a typed $\lambda$-calculus the notion of \emph{type isomorphism}  is a particularisation of the general notion of isomorphism in category theory, with the requirement that the morphisms proving the isomorphism are $\lambda$-definable.
More specifically, two \emph{types} $\tS$ and $\tT$ are \emph{isomorphic} if there are two $\lambda$-terms
$M$ and $N$ of types $\tS\to\tT$ and $\tT\to\tS$, respectively, such that $M\circ N$ is  $\beta\eta$-equal to the identity at type $\tT$ and  $N\circ M$ is  $\beta\eta$-equal to the identity at type $\tS$ ($M\circ N$ is short for $\lambda x. M(Nx)$, where $x$ is fresh).

The importance of type isomorphism has been highlighted by Di Cosmo~\cite{Dicosmo93},
who noted that the equivalence relation on types induced by the notion of isomorphism allows one to abstract from inessential details in the representation of data in programming languages.
To distinguish isomorphic types can entail useless drawbacks; for instance, if a library contains a function of type $\tA\wedge\tB\to\tC$, a request on a function of type $\tB\wedge\tA\to\tC$ will not have success. Note that types as keys are actually used in Hoogle~\cite{mitj}, an Haskell API search engine which allows one to search many standard Haskell libraries by either function name, or by approximate type signature. Neil Mitchell~\cite{mit} remarks
that in this application  a suitable notion of  ``closeness" of types is needed, and isomorphism represents one of the possible meanings of type closeness.
Recently, D\'iaz-Caro and Dowek~\cite{DD14}  pointed out that in typed lambda-calculus, in programming languages, and in proof theory, isomorphic types are often identified. For example, the definitionally equivalent types are identified in Martin-L\"{o}f's type theory and in the Calculus of Constructions.
 For this reason~\cite{DD14} proposes a type system in which $\lambda$-terms getting a type have also all types isomorphic to it.

In the simply typed $\lambda$-calculus, the isomorphism has been characterised by Bruce and Longo~\cite{BruceLongo85}
using
the {\bfseries swap} equation:\quad
$\sigma \rightarrow \tau \rightarrow \rho  ~\iso~  \tau \rightarrow \sigma\rightarrow
    \rho$.
    In
richer $\lambda$-calculi, obtained from the simply typed one by adding other
type constructors (like product types~\cite{Solovev83,BruceDicosmoLongo92,soloviev93complete}) or by allowing higher-order types (System F~\cite{BruceLongo85,Dicosmo93}), the set of equations characterising isomorphic types is obtained in an incremental way.
A survey of these results is given by Di Cosmo in
\cite{MSCSSurvey05}.

As pointed out in~\cite{DDGT10},~\cite{CDMZ13a}, this incremental approach
does not work when intersection and union types are considered. The isomorphism
is no longer a  congruence and that prevents to give it a finitary axiomatisation.
The lack of congruence can be shown considering,
for instance,
the types
\myformula{$\sigma = \varphi_1\to\varphi_2\to\varphi_3$ ~~~and~~~ $\tau = \varphi_2\to\varphi_1\to\varphi_3$.}  They are isomorphic (by argument swapping), while, in general,  both their intersection and their union with another type, for instance $\rho = \varphi_4\to\varphi_5\to\varphi_6$, are not. The reason is that $\sigma$ and $\tau$ are isomorphic by argument swapping, while $\rho$ is isomorphic to itself by identity.

 The standard models of intersection and union types map types to subsets of any domain that is a model of the untyped $\lambda$-calculus, with the conditions that the arrow is interpreted as the function space constructor and the intersection and union operators as the corresponding set-theoretic operators~\cite{barba}.
Oddly enough,
type equality in the standard interpretation of intersection types in $\lambda$-models
does not imply type isomorphism~\cite{DDGT10} and it is so also for union types.
 This fact is due to the interference between atomic types, without functional behaviour, and functional types. For example, $\tA\vee\tB\to\tC$ and $\tB\vee\tA\to\tC$ are equal in all standard models, and isomorphic. In fact, the term $\lambda xy.xy$ has both the types $(\tA\vee\tB\to\tC)\to\tB\vee\tA\to\tC$ and $(\tB\vee\tA\to\tC)\to\tA\vee\tB\to\tC$; note that these isomorphic types are both functional, and this fact is exploited in the deductions. On the contrary, the considered types are no longer
isomorphic when put in intersection or in union
with an atomic type $\varphi$, although their interpretations remain equal;
indeed, there is no $\lambda$-term mapping $(\tB\vee\tA\to\tC)\vee\varphi$ to $(\tA\vee\tB\to\tC)\vee\varphi$, or vice-versa, since when a functional type is put in union (or in intersection) with an atomic type, the possibility of exploiting its functional shape is lost.
Despite these problems,  a characterisation of type isomorphism
is given in~\cite{CDMZ13a}, by defining an (effective) notion of type similarity which turns out to correspond to isomorphism.

The existence of non-isomorphic, but semantically equal, types
reveals a weakness of the type assignment system considered in~\cite{CDMZ13a}, due essentially to the fact that atomic types do not have a  functional behaviour.
This assumption is indeed questionable in the pure $\lambda$-calculus, where everything is a function.
A type system for intersection types in which type isomorphism contains type equality has been proposed in~\cite{CDMZ14} by assuming each atomic type equivalent to a functional one, in such a way that they can be freely interchanged in any deduction.

In the present paper, we extend
the result of~\cite{CDMZ14} considering also union types.
This extension is not trivial owing to the rather odd nature of union types. For instance, as remarked in~\cite{barba}, in systems with intersection and union types, subject reduction does not hold in general.

Following~\cite{CDMZ14}, each atomic type is interpreted as the set of constant functions returning values belonging to the set itself. This is realised
by assuming that any atomic type $\varphi$ is equivalent to $\omega \to \varphi$ (where $\omega$ is the type interpreted as the whole domain).
This choice is motivated by the definition of initial projections in Scott's $D_\infty$ $\lambda$-model~\cite{S72} and from the relations between inverse limit models and filter models~\cite{CDHL84}. In $D_\infty$  each element of the initial domain $D_0$ is projected in a constant function which returns itself when applied to any argument.
As proved in~\cite{CDHL84},  $D_\infty$ is isomorphic to a filter $\lambda$-model built from a set of atomic types 
which correspond to compact elements of the initial domain $D_0$. This model equates $\varphi$ to $\omega \to \varphi$ by construction.
In an applicative setting  it is sensible to assume a semantics in which a constant value (say, an integer), when used as a function, 
returns itself, independently of its argument, validating the present functional interpretation of atomic types.

{\bf Summary} Section~\ref{tas} presents the type assignment system with its properties, notably Subject Reduction and Subject Expansion. Section~\ref{iso} introduces the notion of
isomorphism.
Section~\ref{nos} defines a set of isomorphism preserving normalisation rules for types.
Section~\ref{sse} gives a notion of similarity between types in normal form which assures isomorphism. Section~\ref{conc} draws some possible further work.

\section{Type Assignment System}\labelx{tas}
Let $\Co$ be a denumerable set of atomic types ranged over by $\varphi,\psi$,  and $\tu$ an atom not in $\Co$. The  syntax of  types is given by:
\begin{center}
$\begin{array}{lll}
\tA&::=&\varphi~~\mid~\tu~~\mid~\sigma\to\sigma~\mid~\sigma\wedge\sigma\mid~\sigma\vee\sigma
\end{array}$
\end{center}
As usual, parentheses are omitted according to the precedence rule ``$\wedge$ and $\vee$ over $\rightarrow$'' and ``$\to$'' associates to the right.  Arbitrary types are ranged over by $\tA , \tB  , \tC, \tD$.

The following equivalence asserts the functional character of atomic types, by equating them to arrow types.  It also states that $\tu$ is the top type, viewing intersection and union set-theoretically.

\begin{definition}[Semantic type equivalence]\labelx{eq}
The {\em semantic equivalence relation $\esim$ on types} is defined as the minimal congruence such that :
\myformula{$\varphi  \esim \omega \to \varphi$ \qquad
$\omega \esim \omega \to \omega$ \qquad
 $\tA  \esim \tA \wedge \omega$    \qquad    $\tA  \esim \omega \wedge \tA$ \qquad
 $\omega  \esim \tA \vee \omega$    \qquad    $\omega  \esim \omega \vee \tA$.
}\end{definition}
The congruence allows one to state that $\tA \esim \tA'$ and $\tB \esim \tB'$ imply $\tA \wedge \tB  \esim \tA' \wedge \tB'$ and $\tA \vee \tB  \esim \tA' \vee \tB'$. Moreover  $\tA \to \tB  \esim \tA'  \to \tB'$ if and only if $\tA \esim \tA'$ and $\tB \esim \tB'$.  Note that no other equivalence is assumed between types, for instance $\tA \wedge \tB$ is different from $\tB \wedge\tA$ and $\tA \vee \tB$ is different from $\tB \vee\tA$.

\smallskip

In the type assignment system considered in this paper types are assigned only to
 linear $\lambda$-terms. A $\lambda$-term is {\em linear} if each free or bound variable occurs exactly once in it. This is justified by the observation that type isomorphisms are realised by particular linear $\lambda$-terms, called ``finite hereditary permutators" (see Definitions~\ref{fhp} and~\ref{ti}).
 This is not restrictive since it is easy to prove that the full system, without linearity restriction~\cite{barba}, is conservative over the present one. Therefore the types that can be derived for the finite hereditary permutators  are the same in the two systems, so the present study of type isomorphism holds for
the full system too.

\begin{figure}
\centerline{
$
\begin{array}{ll@{~~~~~~}ll}
(Ax)& \quad\quad   x\dup\tA \vdash x\dup\tA&
(\esim )& \quad \db  \frac{\B
\vdash M\dup\tA\quad \tA \esim \tB}{\B \vdash M\dup \tB}\\
\\
(\to I) &  \db \frac{\B,x\dup\tA \vdash M\dup\tB} {\B \vdash \lambda
x.M\dup\tA \to \tB} & (\to E) & \db \frac{\B_{1} \vdash M\dup\tA \to
\tB \quad
\B_{2} \vdash N\dup\tA}{\B_{1}, \B_{2} \vdash MN\dup\tB} \\
\\
(\wedge I) &  \db \frac{\B \vdash M\dup\tA ~~ \B \vdash M\dup\tB}
{\B \vdash M\dup\tA \wedge \tB} &
(\wedge E) & \db \frac{\B \vdash M\dup\tA \wedge \tB}{\B \vdash M\dup\tA}
~~~~~~\db \frac{\B \vdash M\dup\tA \wedge \tB}{\B \vdash M\dup\tB}\\
\\
\multicolumn{4}{c}{(\vee I) \quad  \db \frac{\B \vdash M\dup\tA}
{\B \vdash M\dup\tA\vee\tB}~~~~~~  \db \frac{\B \vdash M\dup\tA}
{\B \vdash M\dup\tB\vee\tA}}\\\\
\multicolumn{4}{c}{(\vee E) \quad \db \frac{\B_{1}, x\dup\tA\wedge\tD
\vdash M\dup\tC\quad \B_{1}, x\dup\tB\wedge\tD \vdash M\dup\tC\quad
\B_{2} \vdash N\dup(\tA\vee\tB)\wedge\tD}{\B_{1},\B_{2} \vdash
M[N/x]\dup\tC}}\end{array}
$
}
\caption{Typing rules.}\labelx{tr}
\end{figure}
Figure~\ref{tr} gives the typing rules.
 As usual, {\em environments} associate
variables to types and  contain at most one type for each variable.
The environments are relevant, i.e. they contain only the used premises. The domain of the environment $\Gamma$ is denoted by $dom(\Gamma)$. When writing $\Gamma_1,  \Gamma_2$ one convenes
that $dom(\Gamma_1)\cap  dom(\Gamma_2)=\emptyset$. It is easy to verify that $\B \vdash M\dup \tA$ implies $dom(\Gamma)=FV (M)$, where $FV (M)$ denotes the set of free variables of $M$. 

\smallskip

The following rules are admissible.

\smallskip

\centerline{$\begin{array}{ccc}
( L) \quad \db \frac{x\dup\tA \vdash x\dup\tB\quad \B,x\dup\tB \vdash M\dup\tC}{\B,x\dup\tA \vdash M\dup\tC}
&\qquad &(\tu) \quad \db \frac{dom(\Gamma)=FV (M)}{\B \vdash  M\dup\tu}\\[10pt]
(C) ~~ \db \frac{\B_{1},x\dup\tA \vdash M\dup\tB\quad \B_{2} \vdash N\dup\tA}{\B_{1}, \B_{2} \vdash M[N/x]\dup\tB}
&\qquad&(\vee I') ~~ \db \frac{\B, x\dup\tA \vdash M\dup\tC\quad \B, x\dup\tB \vdash M\dup\tC}{\B, x\dup\tA\vee\tB \vdash M\dup\tC}
\\[10pt]
\multicolumn{3}{c}{
(\vee E') ~~ \db \frac{\B_{1}, x\dup\tA \vdash M\dup\tC\quad \B_{1}, x\dup\tB \vdash M\dup\tC\quad \B_{2} \vdash N\dup\tA\vee\tB}{\B_{1},\B_{2} \vdash M[N/x]\dup\tC}}
\end{array}$}
\smallskip
\smallskip
Remark that, considering only linear terms, cut elimination (rule $(C)$) corresponds to standard $\beta$-reduction, while for arbitrary terms  parallel reductions are needed; for details see~\cite{barba}.
Therefore one can state
:
\smallskip
\begin{theorem}[SR] \labelx{srtheorem}
If  $\B \vdash M\dup\tA$ and $M\longrightarrow_{\beta}^*N$, then $\B \vdash
N\dup\tA$.
\end{theorem}

\smallskip
The Subject Reduction Theorem allows one to show some properties useful in the following proofs.

\begin{corollary} \labelx{cor}
\begin{enumerate}
\item\labelx{cor1} If $~\B\vdash \lambda x.M\dup\tA\to\tC$ and $~\B\vdash \lambda x.M\dup\tA\to\tD$, then $~\B\vdash \lambda x.M\dup\tA\to \tC\wedge\tD$.
\item\labelx{cor2} If $~\B\vdash \lambda x.M\dup\tA\to\tB$ and $~\B\vdash \lambda x.M\dup\tC\to\tB$, then $~\B\vdash \lambda x.M\dup\tA\vee\tC\to\tB$.
\item\labelx{cor5} If $~\B\vdash \lambda x.M\dup\tA\to\tC$ and $~\B\vdash \lambda x.M\dup\tB\to\tD$, then $~\B\vdash \lambda x.M\dup\tA\wedge\tB\to \tC\wedge\tD$ and $~\B\vdash \lambda x.M\dup\tA\vee\tB\to\tC\vee\tD$.
\end{enumerate}
\end{corollary}

\smallskip

In the considered system types are not preserved by $\eta$-reduction, as proved by the simple example:
\myformula{$\vdash \lambda xy.xy\dup \varphi\to\psi\to\varphi\text{, but} \not\vdash\lambda x.x \dup \varphi\to\psi\to\varphi$}

\smallskip

On the contrary, subject expansion holds for both $\beta$ and $\eta$-expansions.
\begin{theorem}[Subject Expansion] \labelx{setheorem}
If  $M$ is a linear $\lambda$-term and $M\longrightarrow_{\beta\eta}^*N$ and $\B \vdash
N\dup\tA$, then $\B \vdash M\dup\tA$.
\end{theorem}
\begin{proof}
For $\beta$-expansion it is enough to show: $\B \vdash M[N/x]\dup\tA$ implies  $\B \vdash (\lambda x.M)N\dup\tA$. The proof is by induction on the derivation of $\B \vdash M[N/x]\dup\tA$.
The only interesting case is when the last applied rule is
\myformulaE{$(\vee E) \quad \db \frac{\B_{1}, x\dup\tC\wedge\tD \vdash M\dup\sigma\quad \B_{1}, x\dup\tB\wedge\tD \vdash M\dup\sigma\quad \B_{2} \vdash N\dup(\tC\vee\tB)\wedge\tD}{\B_{1},\B_{2} \vdash M[N/x]\dup\sigma}$}
It is easy to derive $x\dup(\tC\vee\tB)\wedge\tD \vdash x\dup(\tC\wedge\tD)\vee(\tB\wedge\tD)$. Rule $(\vee I')$ applied to the first two premises gives $\B_{1}, x\dup(\tC\wedge\tD)\vee(\tB\wedge\tD)\vdash M\dup\sigma$. So rule $(L)$ derives $\B_{1}, x\dup(\tC\vee\tB)\wedge\tD\vdash M\dup\sigma$, and rule $(\to I)$ derives $\B_{1}\vdash \lambda x. M\dup (\tC\vee\tB)\wedge\tD\to\sigma$. Rule $(\to E)$ gives the conclusion.\\
For $\eta$-expansion the proof is by induction on types. The only interesting case is when $\tA = \tB \vee \tC$. Using rule $(\vee E')$ and applying the induction hypothesis to the first two assumptions one gets:
\myformulaE{     
$\db \frac{ y\dup\tB \vdash \lambda x . y x \dup\tB\quad  y\dup\tC \vdash \lambda x. y x \dup\tC\quad \B \vdash M\dup\tB\vee\tC}{\B \vdash \lambda x. M x\dup\tB \vee \tC}$
   }   \vspace{-10mm}
\end{proof}

\section{Isomorphism} \label{iso}
The study of the type isomorphism in $\lambda$-calculus is based on the characterisation of $\lambda$-term invertibility. A $\lambda$-term $P $ is \textit{invertible} if there exists a
$\lambda$-term $P^{-1}$ such that $P \circ  P^{-1} =_{\beta\eta} P^{-1}\circ P =_{\beta\eta}  = \lambda x.x$. The paper~\cite{Dezani} completely characterises the invertible
$\lambda$-terms in the type free $\lambda \beta \eta$-calculus: the invertible terms are all and only the
\textit{finite hereditary permutators}.
\begin{definition}
 [Finite Hereditary Permutator] \labelx{fhp}A \emph{finite
hereditary permutator} (\fhp\ for short) is a $\lambda$-term of the form (modulo $\beta$-conversion)
\myformula{$\lambda x y_{1}\dots y_{n}.x(P_1 y_{\pi(1)})\dots (P_n y_{\pi (n)}) \; \; \;(n \geq  0)$}
where  $\pi$ is a permutation of $1,\dots,n$, and $P_1,\ldots,P_n$ are \fhp s.
\end{definition}
\noindent
Note that the identity  is trivially an \fhp\ (take $n=0$). Another example of an \fhp\ is
\myformula{$\lambda x y_1 y_2. x\,y_2\,y_1 \mtob \lambda x y_1 y_2. x\,((\lambda z. z)\,y_2)\,((\lambda z. z)\,y_1),$}
which proves the swap equation. 
It is easy to show that \fhp s are closed under composition.

\begin{theorem}
 A $\lambda$-term is invertible iff it
is a finite hereditary permutator.
\end{theorem}

This result,  obtained in the framework of the untyped
$\lambda$-calculus, has been the basis for studying type
isomorphism in different type systems for the $\lambda$-calculus.
Note that
every \fhp\  $P$ has, modulo $\beta \eta$-conversion, a unique
inverse $P^{-1}$.
Even if in the type free $\lambda$-calculus \fhp s are defined modulo $\beta\eta$-conversion~\cite{Dezani}, in this paper \fhp s are considered only modulo $\beta$-conversion, because types are not invariant under $\eta$-reduction.
Taking into account these properties, the definition of type
isomorphism can be
stated as follows:
\begin{definition}[Type Isomorphism]\labelx{ti}
Two types $\sigma$  and  $\tau$ are {\em isomorphic} ($
\sigma \iso \tau$) if there exists a pair $<P,P^{-1}>$ of \fhp s, inverse of each other, such that $\vdash
P\dup \sigma \rightarrow \tau$ and $\vdash P^{-1}
\dup \tau\rightarrow\sigma$. The pair
$<P,P^{-1}>$
\textit{proves} the isomorphism.
\end{definition}
\noindent
When $P=P^{-1}$
one can simply write ``$P$ proves the isomorphism''.

\smallskip

\noindent
It is immediate to verify that type isomorphism is an equivalence relation.

\smallskip

It is useful to single out \fhp s, which only use the identity permutation, and the induced isomorphisms.

\begin{definition} [Finite Hereditary Identity]\labelx{fhi}
A {\em finite hereditary identity} (\fhi) is a $\lambda$-term of the form (modulo $\beta$-conversion)
\myformula{$\lambda x y_{1}\dots y_{n}.x(\id_1 y_{1})\dots (\id_n y_{n}) \; \; \;(n \geq  0)$}
where $\id_1,\ldots,\id_n$ are \fhi s.
\end{definition}
\noindent
The  $\beta$-normal forms of \fhi s are  obtained from the identity $\lambda x.x$ through a finite (possibly zero) number of $\eta$-expansions. Then by Theorem~\ref{setheorem} $\vdash \id\dup\tA\to\tA$ for all \fhi s $\id$ and all $\tA$.

\begin{definition} [Strong Type Isomorphism]\labelx{sti}
Two types $\sigma$  and  $\tau$ are {\em strongly isomorphic} ($
\sigma \isos \tau$) if their isomorphism is proved by an \fhi.
\end{definition}

Notice that
requiring the isomorphism be proved by a pair of \fhi s (instead of a single \fhi) gives an equivalent definition of strong isomorphism, since types are preserved by $\eta$-expansion (Theorem~\ref{setheorem}).

\smallskip

Isomorphism does not imply strong isomorphism, for example $\lambda xyz.xzy$ proves $\tu\to\varphi\to\varphi\iso\varphi\to\varphi$, but $\tu\to\varphi\to\varphi\not\isos\varphi\to\varphi$.
Moreover semantic type equivalence implies strong type isomorphism, i.e. $\tA\esim\tB$ implies $\tA\isos\tB$, but   the inverse does not hold, since $\lambda x.x$ proves $\tA\vee\tB\isos\tB\vee\tA$, but $\tA\vee\tB\not\esim\tB\vee\tA$.

\smallskip

      It is useful to consider some strong isomorphisms, which are directly related to set theoretic properties of intersection and union and to standard properties of functional types. Moreover, all these isomorphisms are provable equalities in the system {\bf B$_+$} of relevant logic~\cite{SE3}.

  \begin{lemma}\labelx{arrowIso} The following strong isomorphisms hold:
\myformula{$\begin{array}{llll}
 \mbox{\rm{\bfseries idem}.}&  \tA \wedge \tA \isos \tA,\ \tA \vee \tA \isos \tA&
     \mbox{\rm{\bfseries comm}.}&  \tA \wedge \tB \isos \tB\wedge \tA,\ \tA \vee \tB \isos \tB\vee \tA\\
  \multicolumn{4}{l}{ \mbox{\rm{\bfseries assoc}.} \quad\quad (\tA \wedge \tB)\wedge \tC \isos \tA\wedge (\tB\wedge \tC),~~ (\tA \vee \tB)\vee \tC \isos \tA\vee (\tB\vee \tC)}\\
  \mbox{\rm{\bfseries dist$\wedge\vee$}.}&   (\tA\vee \tB) \wedge \tC  \isos (\tA\wedge \tC)\vee (\tB \wedge \tC) &
\mbox{\rm{\bfseries dist$\vee\wedge$}.}& (\tA\wedge\tB)\vee\tC \isos (\tA\vee\tC)\wedge(\tB\vee\tC) \\
 \mbox{\rm{\bfseries dist$\to\wedge$}.}& \tA  \to \tB\wedge \tC \isos (\tA \to \tB)  \wedge (\tA \to \tC) &
    \mbox{\rm{\bfseries dist$\to\vee$}.}& \tA\vee \tB \to \tC \isos (\tA \to \tC)  \wedge (\tB \to \tC)
\end{array}$}
\end{lemma}
\begin{proof}
The identity $\lambda x.x$ proves all these isomorphisms except the last two, proved by the $\eta$-expansion of the identity $\lambda xy.xy$.
\end{proof}

As regards to type interpretations, if $\tA$ is included in $\tB$, the intersection $\tA\wedge\tB$ is set-theoretically equal to $\tA$ and the union $\tA\vee\tB$ is set-theoretically equal to $\tB$. So, it is handy to introduce a pre-order on types which formalises set-theoretic inclusion and which takes into account the meaning of the arrow type constructor and the semantic type equivalence given in Definition~\ref{eq}. This pre-order is dubbed normalisation pre-order being used in the next section to define normalisation rules (Definition~\ref{rrt}).
\begin{definition}[Normalisation pre-order on types]\labelx{pt}
The {\em normalisation relation $\leq$ on types} is the minimal pre-order relation such that:
\myformula{$\begin{array}{c}
\tS \leq  \omega  \quad \quad
 \tS \wedge  \tT \leq  \tS \quad \quad \tS \wedge  \tT \leq  \tT \quad \quad \tS \leq  \tS \vee  \tT \quad \quad \tT \leq  \tS \vee  \tT
   \\[1mm]
    %
              \tS \leq  \tT,\; \tS \leq  \tR \Rightarrow
              \tS  \leq  \tT \wedge  \tR \quad \quad \tS \leq  \tT,\; \tR \leq  \tT \Rightarrow
              \tS  \vee  \tR \leq  \tT
              \\[1mm]
\varphi\leq\tS\to\varphi \quad \quad \tu\leq\tS\to\tu\quad \quad \tS' \leq  \tS, \;  \tT \leq  \tT'
             \Rightarrow  \tS\to\tT  \leq  \tS'
             \to  \tT'
 \end{array}$}
\end{definition}
Notice that $\tS \leq  \omega$ agrees with $\tS \wedge  \omega\esim\tS$. Moreover $\varphi\leq\tS\to\varphi$ and $\tu\leq\tS\to\tu$ are justified by $\varphi\esim\tu\to\varphi$, $\omega \esim \omega \to \omega$ and the contra-variance of $\leq$ for arrow types.

\smallskip

The soundness of the normalisation pre-order follows from the following lemma, which shows the expected isomorphisms. To prove this lemma it is useful to observe that  for
each \fhi\  $\id$,  different from the identity, one gets $\id\mtob\lambda xy. \id_1(x(\id_2y))$ for some 
\fhi s $\id_1,\id_2$.
For example,  for  \myformula{$\id$ = $\lambda x y_1 y_2 y_3.x(\lambda t.y_1t) y_2 (\lambda u_1 u_2.y_3 u_1 u_2)$} one has $\id_1 = \lambda x y_2 y_3.x y_2 (\lambda u_1 u_2.y_3 u_1 u_2)$ and $\id_2 = \lambda x t.x t$.

\medskip

 The following lemma proves the validity of two more strong isomorphisms:
\myformula{ \mbox{\rm{\bfseries erase}.}\qquad  if $\tA\leq\tB$ \qquad then \qquad
$\tS \wedge  \tT\isos\tS$ \qquad and \qquad $\tS \vee  \tT\isos\tT$}

\begin{lemma}\labelx{ma}
\begin{enumerate}
\item\labelx{ma2} If $\tS \leq  \tT$, then there is an \fhi\ $\id$ such that $\vdash  \id \dup \tA \to \tB$.
\item\labelx{ma3} If $\tS \leq  \tT$, then $\tS \wedge  \tT\isos\tS$ and $\tS \vee \tT\isos\tT$.
\end{enumerate} \end{lemma}
\begin{proof}

 (\ref{ma2}). The proof is by induction on the definition of $\leq$. Only interesting cases are considered. \\
 In case $\tA\leq\tC$ and $\tC\leq\tB$ imply $\tA\leq\tB$,  by the induction hypothesis there are \fhi s $\id_1$, $\id_2$ such that $\vdash \id_1\dup \tA\to\tC$  and $\vdash \id_2\dup \tC\to\tB$. This implies $\vdash \lambda x. \id_2(\id_1x)\dup \tA\to\tB$. It is easy to verify that $\lambda x. \id_2(\id_1x)$ reduces to an \fhi.\\
 In case $\tA\leq\tB$ and $\tA\leq\tC$ imply $\tA\leq\tB\wedge\tC$, by the induction hypothesis there are \fhi s $\id_1$, $\id_2$ such that $\vdash \id_1\dup \tA\to\tB$  and $\vdash \id_2\dup \tA\to\tC$. By Subject Reduction (Theorem~\ref{srtheorem}) $\vdash \id'_1\dup \tA\to\tB$  and $\vdash \id'_2\dup \tA\to\tC$, where $\id'_1$ and $\id'_2$ are the $\beta$-normal forms of $\id_1$ and $\id_2$, respectively.
 By Subject Expansion (Theorem~\ref{setheorem}) there is an \fhi\ $\id$, $\eta$-expansion of both $\id'_1$ and $\id'_2$,
   such that $\vdash \id\dup \tA\to\tB$  and $\vdash \id\dup \tA\to\tC$; by Corollary~\ref{cor}(\ref{cor1}) $\vdash \id\dup \tA\to\tB\wedge\tC$.
 For the case $\tA\leq\tB$ and $\tC\leq\tB$ imply $\tA\vee\tC \leq \tB$, the proof is similar.\\
 In case $\varphi\leq\tA\to\varphi$, one can derive $y\dup\tS\vdash y\dup\tu$ by rule $(\tu)$, and $x\dup\varphi\vdash x\dup\tu\to\varphi$ by rule $(\esim)$. Then $\vdash\lambda xy. xy\dup\varphi\to\tA\to\varphi$ holds by rules $(\to E)$ and $(\to I)$.\\
 In case $\tA'\leq\tA$ and $\tB\leq\tB'$ imply $\tA\to\tB\leq\tA'\to\tB'$, by the induction hypothesis there are \fhi s $\id_1$, $\id_2$ such that \mbox{$\vdash \id_2\dup \tA'\to\tA$}  and $\vdash \id_1\dup \tB\to\tB'$. This implies $\vdash \lambda xy. \id_1(x(\id_2y))\dup (\tA\to\tB)\to\tA'\to\tB'$. \\
 (\ref{ma3}). By point (\ref{ma2}) there is an \fhi\ $\id$ such that $\vdash  \id \dup \tA \to \tB$. Clearly $\vdash  \id \dup \tA \to \tA$. Corollary~\ref{cor}(\ref{cor1}) gives $\vdash  \id \dup \tA \to \tA\wedge\tB$.
 Since, obviously,  $\vdash  \lambda x.x \dup \tA\wedge\tB \to \tA$, Theorem~\ref{setheorem} assures that $\id$ proves the strong isomorphism $\tA\wedge\tB \isos\tA$. In a similar way one proves that there is an $\id$
  proving $\tA\vee\tB \isos\tB$.

\end{proof}

Strong isomorphism is a congruence, as shown in the following lemma, where
\emph{ type contexts} are defined as usual:
\myformula{$\C{~}~::=~[~] ~\mid ~\C{~}\to\tS~\mid~ \tS\to\C{~}~\mid~ \tS\wedge\C{~}~\mid~\C{~}\wedge\tS ~\mid~ \tS\vee\C{~}~\mid~\C{~}\vee\tS$}

\begin{lemma}\labelx{in}
If $\tA\isos\tB$, then $\C\tA\isos\C\tB$.
\end{lemma}
\begin{proof}
The proof is by structural induction on type contexts. For the empty context it is trivial.
For any other context $\C{~}$,  an \fhi\ $\id_{\C{~}}$ that proves the isomorphism $\C\tA \isos \C\tB$ is given by:
\myformula{$\begin{array}{lll}
 \id_{\C{~}\to\tC}\tob\lambda xy.x(\id_{\C{~}} y)&\qquad\qquad\qquad&
\id_{ \tC\to\C{~}}\tob\lambda xy.\id_{\C{~}}(xy)\\
\id_{\tC \wedge\C{~}} = \id_{\C{~} \wedge \tC} = \id_{\C{~}}&&
\id_{\tC \vee\C{~}} = \id_{\C{~} \vee \tC} = \id_{\C{~}}
\end{array}$}\vspace{-3mm}
\end{proof}

Owing to this lemma, types can be considered modulo idempotence, commutativity and associativity.

%

\section{Normalisation}\labelx{nos}
To investigate type isomorphism, following  a common approach~\cite{BruceDicosmoLongo92,MSCSSurvey05,DDGT10,CDMZ13},
      a notion of \emph{normal form} of types is introduced. {\em Normal type} is short for type in normal form. The notion of normal form is effective, as shown by Theorem~\ref{uni}.

\smallskip

Type normalisation rules are introduced together with the proof of their soundness.

\begin{definition}[Type normalisation rules]\labelx{rrt}
\begin{enumerate}
\item \labelx{rrt1} The inner type normalisation rules are:\\[2pt]
$\hspace*{-8mm}\begin{array}{ll}
(\varphi\rightsquigarrow)\quad
\omega\to\varphi \rightsquigarrow\varphi & \quad
(\tu\rightsquigarrow)\quad
\omega \leq \tS \text{ and } \tS\not=\tu \text{ imply }\tS \rightsquigarrow \omega \\[1mm] 
(\wedge \rightsquigarrow)\quad
\tS\to\tT\wedge\tR \rightsquigarrow(\tS\to\tT)\wedge(\tS\to\tR)& \quad
(\to_\wedge\rightsquigarrow)\quad
(\tA\vee\tT)\wedge\tC \to\tD \rightsquigarrow (\tA\wedge\tC)\vee (\tB \wedge \tC)\to\tD \\[1mm] 
(\vee \rightsquigarrow)\quad
\tS\vee\tT\to\tR\rightsquigarrow(\tS\to\tR)\wedge(\tT\to\tR)
& \quad(\to_\vee\rightsquigarrow)\quad \tS\to (\tB\wedge\tC)\vee\tD \rightsquigarrow \tA\to(\tB\vee\tD)\wedge(\tC\vee\tD)
\\[1mm] 
 \multicolumn{2}{c}{(\leq\rightsquigarrow)\quad
 \tS \leq  \tT\text{ implies }\tS \wedge  \tT\rightsquigarrow\tS \text{ and } \tS \vee  \tT\rightsquigarrow\tT }
\end{array}$
\item \labelx{rrt2} The top type normalisation rules are:
\myformulaB{
$\begin{array}{ll}
 (\text{ctx}\rr)\quad
 \tA \rightsquigarrow  \tB \text{ implies } \C{\tA} \red  \C{\tB} \quad & \quad (\vee\wedge\rr)\quad
 (\tA\wedge\tB) \vee\tC \red (\tA \vee \tC) \wedge (\tB\vee\tC)
\end{array}$
}
\end{enumerate}
\end{definition}
The first two rules follow immediately from semantic type equivalence; moreover, since $\omega\leq\tA\to\omega$, an admissible rule is \hbox{$\tS\to\omega \rightsquigarrow\omega$.} The following four rules correspond to the distribution isomorphisms.  The last rule corresponds to the  \textbf{erase} isomorphism.
Note that in the inner rules $(\to_\wedge\rightsquigarrow)$ and $(\to_\vee\rightsquigarrow)$ the isomorphism \mbox{\rm{\bfseries dist$\wedge\vee$}} is used only 
on the left of an arrow and the isomorphism \mbox{\rm{\bfseries dist$\vee \wedge$}} is used 
on the right of an arrow, respectively. These rules generate  normal forms for arrow types in which the type on the left is an intersection and the type on the right is a union. Moreover the top rule
$(\vee\wedge\rr)$ allows one to define for types a ``conjunctive" normal form.

For example:
\myformula{$ (\varphi_1\to\varphi_2)\wedge\varphi_2\rightsquigarrow\varphi_2\text{ which implies }
((\varphi_1\to\varphi_2)\wedge\varphi_2)\vee\varphi_3\red\varphi_2\vee\varphi_3$}
\myformula{$(\varphi_1\to\varphi_2)\vee\varphi_2\rightsquigarrow\varphi_1\to\varphi_2\text{ which implies }
((\varphi_1\to\varphi_2)\vee\varphi_2)\wedge\varphi_3\red(\varphi_1\to\varphi_2)\wedge\varphi_3$}

\smallskip

Having two kinds of normalisation rules (inner and top) allows to apply only one of  the isomorphisms  \mbox{\rm{\bfseries dist$\wedge\vee$}} and \mbox{\rm{\bfseries dist$\vee \wedge$}} at each subtype of a type. This is crucial to assure termination of normalisation.

The present normalisation rules are much simpler than those in~\cite{CDMZ13a}. The functional behaviour of atomic types produces this simplification.
\begin{theorem}[Soundness of the normalisation rules]\labelx{srr}
\begin{enumerate}
\item \labelx{srr1}If $\tS \rightsquigarrow \tT$, then $\tS \isos \tT$.
\item \labelx{srr2}If $\tS \red \tT$, then $\tS \isos \tT$.
\end{enumerate}
\end{theorem}
\begin{proof}
(\ref{srr1}). Rule $(\varphi\rightsquigarrow)$ is obtained by orienting the equivalence relation between types, so it is sound since equivalent types are isomorphic. Rule $(\tu\rightsquigarrow)$ is sound because, by Lemma~\ref{ma}(\ref{ma2}), there is  an \fhi    $\;\id$  such that $\der{}\id{\omega \to  \tS}$, and obviously $\der{}{\id}{\tS \to  \omega}$.  Rules $(\wedge\rightsquigarrow)$,   $(\to_\wedge\rightsquigarrow)$, $(\vee\rightsquigarrow)$ and $(\to_\vee\rightsquigarrow)$  are sound by the strong isomorphisms 
of Lemma~\ref{arrowIso}.
 Lemma~\ref{ma}(\ref{ma3}) implies the soundness of rule $(\leq\rightsquigarrow)$. \\
(\ref{srr2}). The soundness of the rule $(\text{ctx}\rr)$ is proved in Lemma~\ref{in}. The strong isomorphism \mbox{\rm{\bfseries dist$\vee\wedge$}} gives the soundness of rule $(\vee\wedge\rr)$.

\end{proof}
For example $((\varphi_1\to\varphi_2)\wedge\varphi_2)\vee\varphi_3\red\varphi_2\vee\varphi_3$, as shown before, and  $\lambda xy.xy$ proves \myformula{$((\varphi_1\to\varphi_2)\wedge\varphi_2)\vee\varphi_3\isos\varphi_2\vee\varphi_3$.}

\smallskip

The following theorem shows the existence and uniqueness of the normal forms, i.e. that the top normalisation rules are terminating and confluent.

\begin{theorem} [Uniqueness of normal form]\labelx{uni}
The top normalisation rules of Definition~\ref{rrt} are terminating and confluent.
\end{theorem}
\begin{proof}The {\em termination} follows from an easy adaptation of the recursive path ordering method~\cite{D82}. The partial order on operators is defined by: $\to~\succ~\vee~\succ~\wedge$ for  holes at top level or in the right-hand-sides of arrow types and  $\to~\succ~\wedge~\succ~\vee$ for holes in the left-hand-sides of arrow types. Notice that the induced recursive path ordering $\succ^*$ has the subterm property. This solves the case of rules $(\varphi\rightsquigarrow)$, $(\tu\rightsquigarrow)$, $(\leq\rightsquigarrow)$.
For
rule $(\wedge \rightsquigarrow)$, since $\to~\succ~\wedge$, it is enough to observe that $\tS\to\tT\wedge\tR~\succ^*~\tS\to\tT$ and        $\tS\to\tT\wedge\tR~\succ^*~\tS\to\tR$.
For rules $(\to_\wedge\rightsquigarrow)$ and $(\vee\wedge\rr)$, since $\vee~\succ~\wedge$ for holes at top level or in the right-hand-sides of arrow types, it is enough to observe that $(\tS\wedge\tT)\vee\tR~\succ^*~\tS\vee\tR$ and        $(\tS\wedge\tT)\vee\tR~\succ^*~\tT\vee\tR$. The proof for the remaining rules are similar.

 For {\em confluence}, thanks to the Newman Lemma~\cite{N42}, it is sufficient to  prove  the convergence of the critical
 pairs. For example, the  types $\tA\vee\tB\vee\tC\to\tD,$ and $\tA\to\tB\wedge\tC\wedge\tD$ give rise 
 to critical pairs, as well as the following ones, when $\tA\leq\tB$:
 \myformula{ $(\tA\wedge\tB)\vee \tC $,\qquad $ \tC \to
(\tA\wedge\tB) \vee \tD$,\qquad $(\tA\vee\tB) \wedge \tC \to \tD$, \qquad
$ \tC \to
(\tA \vee\tB)\wedge\tD$,\qquad$\tC \vee (\tA\wedge\tB)\to \tD$.
}
Other examples of critical pairs are
$(\omega\to\varphi)\wedge\sigma$ if $\omega\to\varphi\leq\sigma$, and $\sigma\vee\omega$ if $\omega\leq\sigma.
$
\end{proof}

The
normal form of a type $\tS$, unique 
modulo commutativity and associativity, is denoted by $\nf{\tS}$. The soundness of the normalisation rules (Theorem~\ref{srr}) implies that each type is strongly isomorphic to its normal form.
\begin{corollary}\label{inf} $\tA\isos\nf\tA$.
\end{corollary}

As expected, semantic equivalent types have the same normal form. Clearly the inverse is false, since, for example, $\nf{(\tA\to\tB\wedge\tC)}=(\tA\to\tB)\wedge(\tA\to\tC)$, but $\tA\to\tB\wedge\tC\not\esim(\tA\to\tB)\wedge(\tA\to\tC)$.

\begin{lemma}\labelx{nc}
If  $\tS\esim \tB$, then $\nf{\tS} = \nf{\tB}$.
\end{lemma}
\begin{proof} The proof is by cases on Definition~\ref{eq}. For the equivalences $\varphi\esim\omega\to\varphi$ and $\omega\esim\omega\to\omega$, rules $(\varphi\rightsquigarrow)$ and $(\omega\rightsquigarrow)$ give $\nf{(\omega\to\varphi)}=\varphi$ and $\nf{(\omega\to\omega)}=\omega$, respectively. For the equivalences $\tS\esim\omega\wedge\tS$, $\tS\esim\tS\wedge\omega$, $\omega\esim\omega\vee\tS$ and $\omega\esim\tS\vee\omega$, rule $(\leq\rightsquigarrow)$, with $\tA\leq\tu$, gives $\nf{(\omega\wedge\tS)}=\nf{(\tS\wedge\omega)}=\tS$ and $\nf{(\omega\vee\tS)}=\nf{(\tS\vee\omega)}=\omega$. The congruence follows from rule $(\text{ctx}\rr)$.
\end{proof}

\section{Similarity as Isomorphism}\labelx{sse}
This section shows the main result of the paper, i.e. that two types with ``similar'' normal forms (Definition \ref{simil}) are isomorphic.
The basic aim of the similarity relation is that of formalising
isomorphism determined by argument permutations (as in the swap
equation). This relation has  to take into
account the fact that, for two types to be isomorphic, it is not
sufficient that they coincide modulo permutation of types in the
arrow sequences, as in the case of cartesian products. Indeed the
same permutation must be applicable to all types
in the corresponding intersections and unions. The key notion of similarity exactly expresses such a condition.

To define similarity, it is useful to distinguish between different kinds of types. So in the following:
\begin{itemize}
\item $\tO, \tP $  range over atomic and normal arrow types, i.e. $\tO ::=\omega\mid \varphi \mid \te \to \tM$;
\item $\te,\ti$ range over normal intersections of atomic and arrow types, i.e. $\te~::=\tO~\mid~\nf{(\te\wedge\te)}$;
\item $\tM,\tN$  range over normal unions of atomic and arrow types, i.e. $\tM~::=\tO~\mid~\nf{(\tM\vee\tM)}$;
\item $\tQ,\tY$ range over normal types, i.e. $\tQ~::=\tM~\mid~\nf{(\tQ\wedge\tQ)}$.
\end{itemize}

\begin{definition}[Similarity]\labelx{simil}
The {\em similarity} relation between two sequences of normal types $\seq{\tQ_1}{\tQ_m}$ and $\seq{\tY_1}{\tY_m}$, written $\seq{\tQ_1}{\tQ_m}\tsi \seq{\tY_1}{\tY_m}$,
is the smallest equivalence relation such that:
\begin{enumerate}
\item\labelx{simil1} $\seq{\tQ_1}{\tQ_m} \tsi \seq{\tQ_1}{\tQ_m}$ 
\item\labelx{simil2} if $\langle{\tQ_1},\ldots,\tQ_i,\tQ_{i+1},\ldots,{\tQ_m}\rangle\tsi \langle{\tY_1},\ldots,\tY_i,\tY_{i+1},\ldots,{\tY_m}\rangle$, then
\myformula{
$\langle{\tQ_1},\ldots,\nf{(\tQ_i\wedge\tQ_{i+1})},\ldots,{\tQ_m}\rangle\tsi \langle{\tY_1},\ldots,\nf{(\tY_i\wedge\tY_{i+1})},\ldots,{\tY_m}\rangle$ and}
\myformula{$\langle{\tQ_1},\ldots,\nf{(\tQ_i\vee\tQ_{i+1})},\ldots,{\tQ_m}\rangle\tsi \langle{\tY_1},\ldots,\nf{(\tY_i\vee\tY_{i+1})},\ldots,{\tY_m}\rangle$;}
\item\labelx{simil3} if $\seq{\te^{(1)}_i}{\te^{(m)}_i} \tsi \seq{\ti^{(1)}_i}{\ti^{(m)}_i}$ for $1\leq i\leq n$ and $\seq{\tM_1}{\tM_m}\tsi\seq{\tN_1}{\tN_m}$, then
\myformula{$\begin{array}{l}\seq{\nf{(\te^{(1)}_1 \to \ldots \to \te^{(1)}_n\to \tM_1)}}{\nf{(\te^{(m)}_1 \to \ldots \to \te^{(m)}_n\to \tM_m)}}\tsi\\[2pt] \seq{\nf{(\ti^{(1)}_{\pi(1)} \to \ldots \to \ti^{(1)}_{\pi(n)}\to \tN_1)}}{\nf{(\ti^{(m)}_{\pi(1)} \to \ldots \to \ti^{(m)}_{\pi(n)}\to \tN_m)}},\end{array}$}
where $\pi$ is a permutation of $1,\dots, n$.
\end{enumerate}
\emph{Similarity between normal types} is trivially defined as similarity between unary sequences:
$\tQ\tsi\tY$ if $\seqs\tQ\tsi\seqs\tY$.
\end{definition}
 \noindent
 For example, $ \seqs{\varphi_1,\tu}\tsi\seqs{\varphi_1,\tu}$, $ \seqs{\tu,\varphi_2}\tsi\seqs{\tu,\varphi_2}$, $ \seqs{\varphi_3,\varphi_4}\tsi\seqs{\varphi_3,\varphi_4}$ imply
 \myformula{$ \seqs{\varphi_1\to\varphi_3,\tu\to\varphi_2\to\varphi_4}\tsi\seqs{\tu\to\varphi_1\to\varphi_3,\varphi_2\to\varphi_4}$ by (\ref{simil3}) and then }
 \myformula{$ \seqs{(\varphi_1\to\varphi_3)\vee(\tu\to\varphi_2\to\varphi_4)}\tsi\seqs{(\tu\to\varphi_1\to\varphi_3)\vee(\varphi_2\to\varphi_4)}$ by (\ref{simil2})}. This, together with  $\seqs{\varphi_5}\tsi\seqs{\varphi_5}$,   $\seqs{\varphi_6}\tsi\seqs{\varphi_6}$, $\seqs{\varphi_7}\tsi\seqs{\varphi_7}$, gives \myformula{$ \seqs{\varphi_5\to\varphi_6\to\varphi_7\to(\varphi_1\to\varphi_3)\vee(\tu\to\varphi_2\to\varphi_4)}\tsi\seqs{\varphi_7\to\varphi_5\to\varphi_6\to(\tu\to\varphi_1\to\varphi_3)\vee(\varphi_2\to\varphi_4)}$} by  (\ref{simil3}).

\smallskip

The proof of the similarity soundness requires some ingenuity.
\begin{theorem}[Soundness]\labelx{ParteVera} If $\seq{\tQ_1}{\tQ_m}\tsi \seq{\tY_1}{\tY_m}$, then
there is a pair of \fhp s 
that proves $\tQ_j\iso\tY_j$, for
$1\leq j\leq m$.
\end{theorem}
\begin{proof}
By induction on the definition of
$\tsi$ (Definition \ref{simil}). 

(\ref{simil1}). $\seq{\tQ_1}{\tQ_m}\tsi \seq{\tQ_1}{\tQ_m}$. The identity proves the isomorphism.

(\ref{simil2}).
$\langle{\tQ_1},\ldots,\nf{(\tQ_i\wedge\tQ_{i+1})},\ldots,{\tQ_m}\rangle\tsi \langle{\tY_1},\ldots,\nf{(\tY_i\wedge\tY_{i+1})},\ldots,{\tY_m}\rangle$~~ since ~~
$\langle{\tQ_1},\ldots,\tQ_i,\tQ_{i+1},\ldots,{\tQ_m}\rangle\tsi $ $ \langle{\tY_1},\ldots,\tY_i,\tY_{i+1},\ldots,{\tY_m}\rangle$.
By the
induction hypothesis there is a pair $<P, P^{-1}\!\!>$ that proves $\tQ_{j}\iso\tY_{j}$, for
$1\leq j\leq m$. By Corollary \ref{cor}(\ref{cor5}), the same pair proves $\tQ_i\wedge\tQ_{i+1}\iso\tY_i\wedge\tY_{i+1}$. By Theorem \ref{srr} there are \fhi s $\id_1,\id_2$ such that
$\id_1$ proves $\tQ_i\wedge\tQ_{i+1}\iso\nf{(\tQ_i\wedge\tQ_{i+1})}$ and $\id_2$ proves $\tY_i\wedge\tY_{i+1}\iso\nf{(\tY_i\wedge\tY_{i+1})}$. Clearly 
\mbox{$\vdash\id_1:\tQ_j\to\tQ_j$} and $\vdash\id_2:\tY_j\to\tY_j$ for $1\leq j\leq m$.
Then the pair \mbox{$<\lambda x.\id_2(P(\id_1x)), \lambda x.\id_1(P^{-1}(\id_2x))\!>$} proves the required isomorphisms.
The proof for the case\\ \centerline{$\langle{\tQ_1},\ldots,\nf{(\tQ_i\vee\tQ_{i+1})},\ldots,{\tQ_m}\rangle\tsi \langle{\tY_1},\ldots,\nf{(\tY_i\vee\tY_{i+1})},\ldots,{\tY_m}\rangle$,} since
$\langle{\tQ_1},\ldots,\tQ_i,\tQ_{i+1},\ldots,{\tQ_m}\rangle\tsi \langle{\tY_1},\ldots,\tY_i,\tY_{i+1},\ldots,{\tY_m}\rangle$, is analogous.\\

(\ref{simil3}).
$$\seq{\nf{(\te^{(1)}_1 \to \ldots \to \te^{(1)}_n\to
\tM_1)}}{\nf{(\te^{(m)}_1 \to \ldots \to \te^{(m)}_n\to \tM_m)}}\tsi $$
\vspace{-5mm}
$$\seq{\nf{(\ti^{(1)}_{\pi(1)} \to \ldots \to \ti^{(1)}_{\pi(n)}\to
\tN_1)}}{\nf{(\ti^{(m)}_{\pi(1)} \to \ldots \to \ti^{(m)}_{\pi(n)}\to
\tN_m)}}$$
 since $\seq{\te^{(1)}_i}{\te^{(m)}_i} \tsi \seq{\ti^{(1)}_i}{\ti^{(m)}_i}$ for $1\leq i\leq n$ and $\seq{\tM_1}{\tM_m}\tsi\seq{\tN_1}{\tN_m}$.
 By the induction hypothesis, there are pairs $<P_i,P^{-1}_{i}\!\!>$ proving $\te_{i}^{(j)}\iso\ti_{i}^{(j)}$ and a pair $<P_*,P^{-1}_*\!\!>$ proving $\tM_j\iso\tN_j$ for $1\leq i\leq n$  and $1\leq j\leq m$. Let\\
  \centerline{$\begin{array}{lll}
P&=&\lambda x y_{1}\ldots y_{n}.(P_*(x  (P^{-1}_{1} y_{\pi^{-1}(1)})\ldots (P^{-1}_{n} y_{\pi^{-1}(n)})))\\
P^{-1}&=&
\lambda x y_{1}\ldots y_{n}. (P_*^{-1}(x (P_{\pi(1)} y_{\pi(1)})\ldots (P_{\pi(n)} y_{\pi(n)})))
\end{array}$}
 It is easy to verify that
  \myformula{$\begin{array}{c}
 \vdash P\dup (\te^{(j)}_1 \to \ldots \to \te^{(j)}_n\to
\tM_j)\to \ti^{(j)}_{\pi(1)} \to \ldots \to \ti^{(j)}_{\pi(n)}\to
\tN_j
\\
\vdash P^{-1}\dup
 (\ti^{(j)}_{\pi(1)} \to \ldots \to \ti^{(j)}_{\pi(n)}\to
\tN_j)\to
 \te^{(j)}_1 \to \ldots \to \te^{(j)}_n\to
\tM_j
 \end{array}$}
 for $1\leq j\leq m$. Notice that $\nf{(\te_1\to\ldots \to\te_h\to\tM)}=\begin{cases}
\te_1\to\ldots \to\te_k\to\tM      & \text{if } \te_{k}\not=\tu\text{ and }\te_{k+1}=\ldots=\te_h=\tu\\
&\text{ and } \tM \text{ is an atomic type}, \\
 \te_1\to\ldots \to\te_h\to\tM     & \text{otherwise}
\end{cases}$\\[3pt]
since $\te_1,\ldots ,\te_h$ are normal intersections of atomic and arrow types and $\tM$ is a normal union of atomic and arrow types. Then $\te_1\to\ldots \to\te_h\to\tM\esim\nf{(\te_1\to\ldots \to\te_h\to\tM)}$, and,
by the typing rule $(\esim)$:
 \myformula{$\begin{array}{c}
 \vdash P\dup \nf{(\te^{(j)}_1 \to \ldots \to \te^{(j)}_n\to
\tM_j)}\to \nf{(\ti^{(j)}_{\pi(1)} \to \ldots \to \ti^{(j)}_{\pi(n)}\to
\tN_j)}
\\
\vdash P^{-1}\dup
 \nf{(\ti^{(j)}_{\pi(1)} \to \ldots \to \ti^{(j)}_{\pi(n)}\to
\tN_j)}\to
 \nf{(\te^{(j)}_1 \to \ldots \to \te^{(j)}_n\to
\tM_j)}
 \end{array}$}
 for $1\leq j\leq m$.
 So $<P,P^{-1}>$ is the required pair.

\end{proof}

An immediate implication of the Soundness Theorem and of Corollary~\ref{inf} is that two types  with similar normal forms are isomorphic.

\begin{corollary} \label{sum}
If $\nf\tA \tsi \nf\tB$, then $\tA \iso \tB$.
\end {corollary}

For example the isomorphism of the types, shown similar after Definition~\ref{simil}, is proved by
\myformula{$<\lambda xy_1y_2y_3y_4y_5. xy_3y_1y_2y_5y_4, \lambda x y_1y_2y_3y_4y_5. xy_2y_3y_1y_5y_4>$.}

\section{Conclusion}\label{conc}

This paper studies type isomorphism for a typed $\lambda$-calculus with intersection and union types, in which all types have a functional character.
 Atomic types become types of functions by assuming
 an equivalence relation that equates any atomic type $\varphi$ to  $\tu\to\varphi$. This equivalence has been introduced in~\cite{CDHL84} for constructing a filter model isomorphic to Scott's $D_\infty$ and it is validated by the standard interpretation of types in this model. In the so obtained type system all types which are set-theoretically equal (using idempotence, commutativity, associativity and distributivity of intersection and union) are proved isomorphic
by the identity.

Basic notions for the given development are those of type normalisation and similarity between normal types. Similarity provides a remarkable insight on isomorphism and we conjecture that, indeed, it gives a complete characterisation of type isomorphism for the system considered in the paper. We leave the proof of this conjecture as future work.

Following  D\'iaz-Caro and Dowek~\cite{DD14} we aim to extend the type assignment systems developed in~\cite{CDMZ13}  and~\cite{CDMZ13a}, by equating all isomorphic types.  This would lead to introduce equivalence rules on $\lambda$-terms, see~\cite{DD14}.
Lastly we plan to study type isomorphism in other assignment system with intersection and union types as, for instance, the ones for the lazy $\lambda$-calculus.\\

\smallskip

\noindent
{\bf Acknowledgements} The authors gratefully thank the 
referees  and Alejandro D\'iaz-Caro for their numerous constructive remarks.

\bibliographystyle{eptcs}
\bibliography{biblio}

\end{document}

%% file: macroL.tex
\newif\ifmm
\mmtrue 

\newif\ifmc
\mcfalse 
\mctrue 

\newtheorem{fact}{Fact}[section]

\newtheorem{lemma}[fact]{Lemma}

\newtheorem{corollary}[fact]{Corollary}

\newtheorem{definition}[fact]{Definition}

\newtheorem{theorem}[fact]{Theorem}

\newcommand{\dup}{\! : \!}
\newcommand{\fhp}{\text{\small FHP}}   
\newcommand{\fhi}{\text{\small FHI}}   

\newcommand{\tA}{\sigma}       
\newcommand{\tB}{\tau}
\newcommand{\tC}{\rho}
\newcommand{\tD}{\zeta}

\newcommand{\tM}{\mu}
\newcommand{\tN}{\nu}

\newcommand{\te}{\xi}

\newcommand{\ti}{\chi}

\newcommand{\tu}{\omega}

\newcommand{\tO}{\alpha}
\newcommand{\tP}{\beta}
\newcommand{\tQ}{\eta}
\newcommand{\tY}{\theta}

\newcommand{\B}{\Gamma}

\newcommand{\db}{\displaystyle}

\newcommand{\seq}[2]{\langle#1,\ldots,#2\rangle}
\newcommand{\seqs}[1]{\langle#1\rangle}
\newcommand{\tsi}{\sim}

\newcommand{\tS}{\sigma}       
\newcommand{\tT}{\tau}
\newcommand{\tR}{\rho}

\newcommand{\der}[3]{#1\vdash #2\dup #3}

\newcommand{\red}{\Longrightarrow}

\newcommand{\C}[1]{{\cal C}[#1]}

\newcommand{\iso}{\approx}

\newcommand{\id}{{\sf Id}}
\newcommand{\tob}{\;{_\beta}\!\!\longleftarrow}
\newcommand{\mtob}{\;^*\!\!\!\!{_\beta}\!\!\longleftarrow}

\newcommand{\labelx}[1]{\label{#1}}

\newcommand{\isos}{\approx_{\sf s}}

\newcommand{\nf}[1]{#1\!\!\downarrow}

\newcommand{\myformula}[1]{\\[0.5pt]\centerline{#1}}

\newcommand{\myformulaE}[1]{\\[4pt]\centerline{#1}\\[4pt]}

\newcommand{\myformulaB}[1]{\\[2pt]\centerline{#1}}

\newcommand{\rr}{\Rightarrow}

\newcommand{\Co}{\texttt{A}}
\newcommand{\esim}{\cong}